\documentclass[a4paper,11pt]{amsart}  
\usepackage{amsmath,amssymb,amsthm,graphicx,subfigure,enumitem,verbatim,fullpage}
\usepackage[left=1.5in,right=1.5in,top=1.4in,bottom=1.4in]{geometry}
\usepackage[pdftitle={Making Octants Colorful and Related Covering Decomposition Problems},
            pdfauthor={J. Cardinal, K. Knauer, P. Micek, T. Ueckerdt}]{hyperref}

\newtheorem{theorem}{Theorem}
\newtheorem{corollary}{Corollary}
\newtheorem{lemma}{Lemma}

\newcommand{\bigp}{\mathcal{P}}

\newcommand{\bign}{\mathcal{N}}
\renewcommand{\bigl}{\mathcal{L}}

\renewcommand{\S}{\mathcal{S}}

\let\leq\leqslant
\let\geq\geqslant
\let\epsilon\varepsilon
\newenvironment{enumeratei}{\begin{enumerate}[label=\textup{(\roman*)}, noitemsep, topsep=1.5mm, labelindent=.8em, leftmargin=*, widest=.]}{\end{enumerate}}

\makeatletter
\let\old@setaddresses\@setaddresses
\def\@setaddresses{\bgroup\parindent 0pt\let\scshape\relax\old@setaddresses\egroup}
\makeatother

\title{Making Octants Colorful\\ 
and Related Covering Decomposition Problems}

\author[J.~Cardinal]{Jean Cardinal}
\author[K.~Knauer]{Kolja Knauer}
\author[P.~Micek]{Piotr Micek}
\author[T.~Ueckerdt]{Torsten Ueckerdt}

\thanks{A preliminary version of this paper was presented at the \emph{2014 ACM-SIAM Symposium on Discrete Algorithms (SODA'14)}.}
\thanks{Journal version of this paper is submitted to \emph{SIAM J.\ Discrete Math}.}

\address[J.~Cardinal]{Universit\'e libre de Bruxelles, Brussels, Belgium}
\email{jcardin@ulb.ac.be}

\address[K.~Knauer]{LIRMM, Universit\'e Montpellier 2, Montpellier, France}
\email{kolja.knauer@gmail.com}

\address[P.~Micek]{Theoretical Computer Science Department, Faculty of Mathematics and Computer Science, Jagiellonian University, Krak\'{o}w, Poland}
\email{piotr.micek@tcs.uj.edu.pl}

\address[T.~Ueckerdt]{Karlsruhe Institute of Technology, Department of Mathematics, Karlsruhe, Germany}
\email{torsten.ueckerdt@kit.edu}

\begin{document}
\date{}
\maketitle

\begin{abstract}
We give new positive results on the long-standing open problem of geometric covering decomposition for homothetic polygons. In particular, we prove that for any positive integer $k$, every finite set of points in $\mathbb{R}^3$ can be colored with $k$ colors so that every translate of the negative octant containing at least $k^6$ points contains at least one of each color. The best previously known bound was doubly exponential in $k$. This yields, among other corollaries, the first polynomial bound for the decomposability of multiple coverings by homothetic triangles. We also investigate related decomposition problems involving intervals appearing on a line. We prove that no algorithm can dynamically maintain a decomposition of a multiple covering by intervals under insertion of new intervals, even in a {\em semi-online} model, in which some coloring decisions can be delayed. This implies that a wide range of sweeping plane algorithms cannot guarantee any bound even for special cases of the octant problem.
\end{abstract}

\section{Introduction and Main Results}

We study coloring problems for hypergraphs induced by simple geometric objects. Given a family of convex bodies in $\mathbb{R}^d$, a natural colorability question that one may consider is the following: is it true that for any positive integer $k$, every collection of points $\bigp \subset\mathbb{R}^d$ can be colored with $k$ colors so that any element of the family containing at least $p(k)$ of them, for some function $p(k)$, contains at least one of each color? This question has been investigated previously for convex bodies in the plane such as halfplanes and translates of a convex polygon. 

\subsection{Octants in three-space.}

In this paper, we give a polynomial upper bound on $p(k)$ when the family under consideration is the set of translates of the three-dimensional negative octant $\{(x,y,z)\in \mathbb{R}^3 : x\leq 0, y\leq 0, z\leq 0\}$. The best previously known bound is due to Keszegh and P{\'a}lv{\"o}lgyi, and is doubly exponential in $k$~\cite{KP14}. 

\begin{theorem}
\label{thm:ub}
There exists a constant $a<6$ such that for any positive integer $k$, every finite set $\bigp$ of points in $\mathbb{R}^3$ can be colored with $k$ colors so that every translate of the negative octant containing at least $k^a$ points of $\bigp$ contains at least one of each color.   
\end{theorem}

A \emph{dual} version of the above problem sometimes referred to as {\em cover-de\-com\-po\-sa\-bi\-li\-ty} can be stated as follows: Given a collection $C$ of convex bodies, we wish to color them with $k$ colors so that any point of $\mathbb{R}^d$ covered by at least $p(k)$ of them, for some function $p(k)$, is covered by at least one of each color. In the \emph{primal} setting with respect to octants we can replace the point set $\bigp$ with a set $C$ of \emph{positive} octants with apices in $\bigp$. Then the primal value of $\bigp$ coincides with the dual value of $C$. Since clearly the dual problem is equivalent if we pick negative instead of positive octants, we have:

\begin{corollary}
\label{cor:dualoct}
There exists a constant $a<6$ such that for any positive integer $k$, every finite set $\bigp$ of translates of the negative octant can be colored with $k$ colors so that every point of $\mathbb{R}^3$ contained in at least $k^a$ octants of $\bigp$ is contained in at least one of each color.   
\end{corollary}

The next corollary is obtained by observing that the intersections of a set of octants with a plane in $\mathbb{R}^3$ that is not parallel to any axis form a set of homothetic triangles (see Figure~\ref{fig:tri}). 
\begin{corollary}
\label{cor:dualtri}
There exists a constant $a<6$ such that for any positive integer $k$, every finite set $\bigp$ of homothetic triangles in the plane can be colored with $k$ colors so that every point contained in at least $k^a$ triangles of $\bigp$ is contained in at least one of each color.   
\end{corollary}

Finally, using standard arguments, the latter result can be extended to infinite sets, and cast as a cover-decomposability statement. Here a covering is said to be {\em decomposable} into $k$ coverings when the objects in the covering can be colored with $k$ colors so that every color class is a covering by itself.
\begin{corollary}
\label{cor:inf}
There exists a constant $a<6$ such that for any positive integer $k$, every locally finite $k^a$-fold covering of the plane by homothetic triangles is decomposable into $k$ coverings.    
\end{corollary}

The proof of Theorem~\ref{thm:ub} is given in Section~\ref{sec:pf1}.\\

\begin{figure}
\begin{center}
\subfigure[\label{fig:tri}From octants to triangles.]{\includegraphics[width=.5\textwidth]{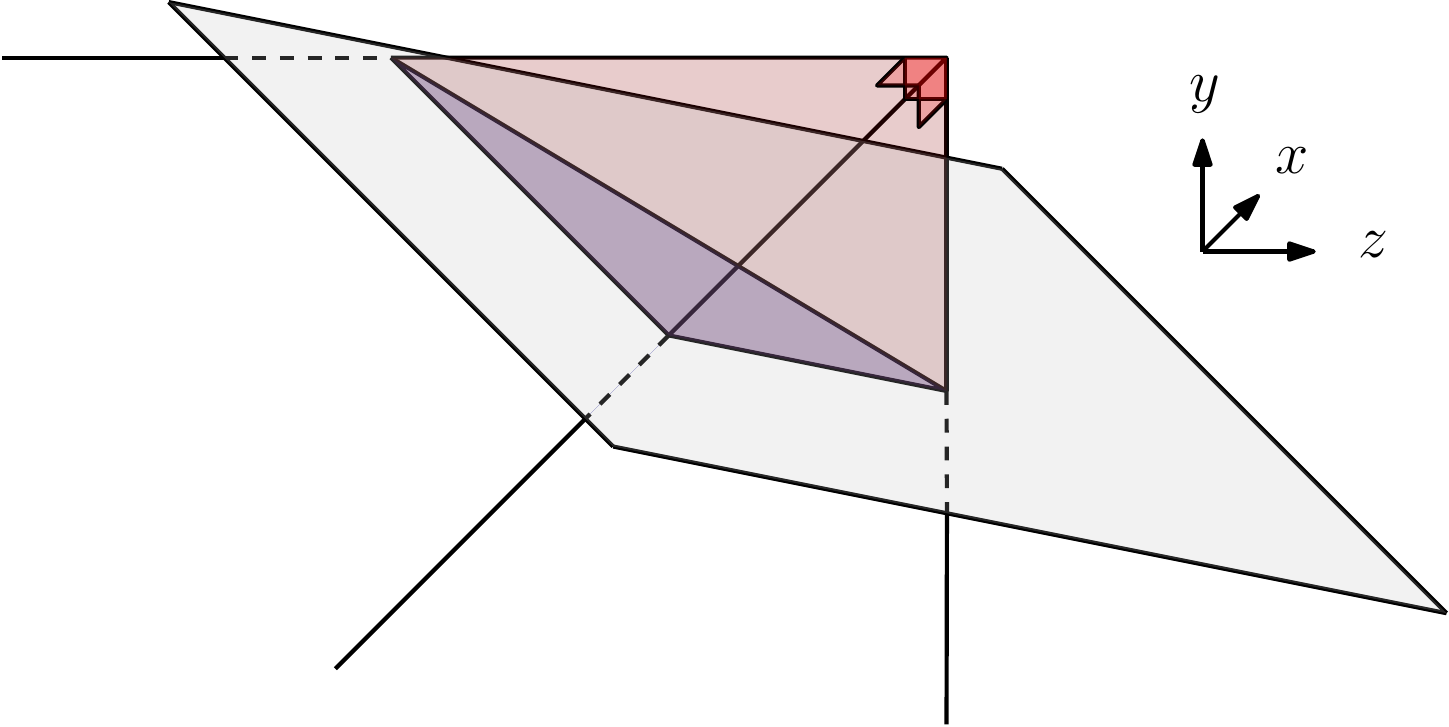}}
\hspace{1cm}
\subfigure[\label{fig:br}From octants to bottomless rectangles.]{\includegraphics[width=.3\textwidth]{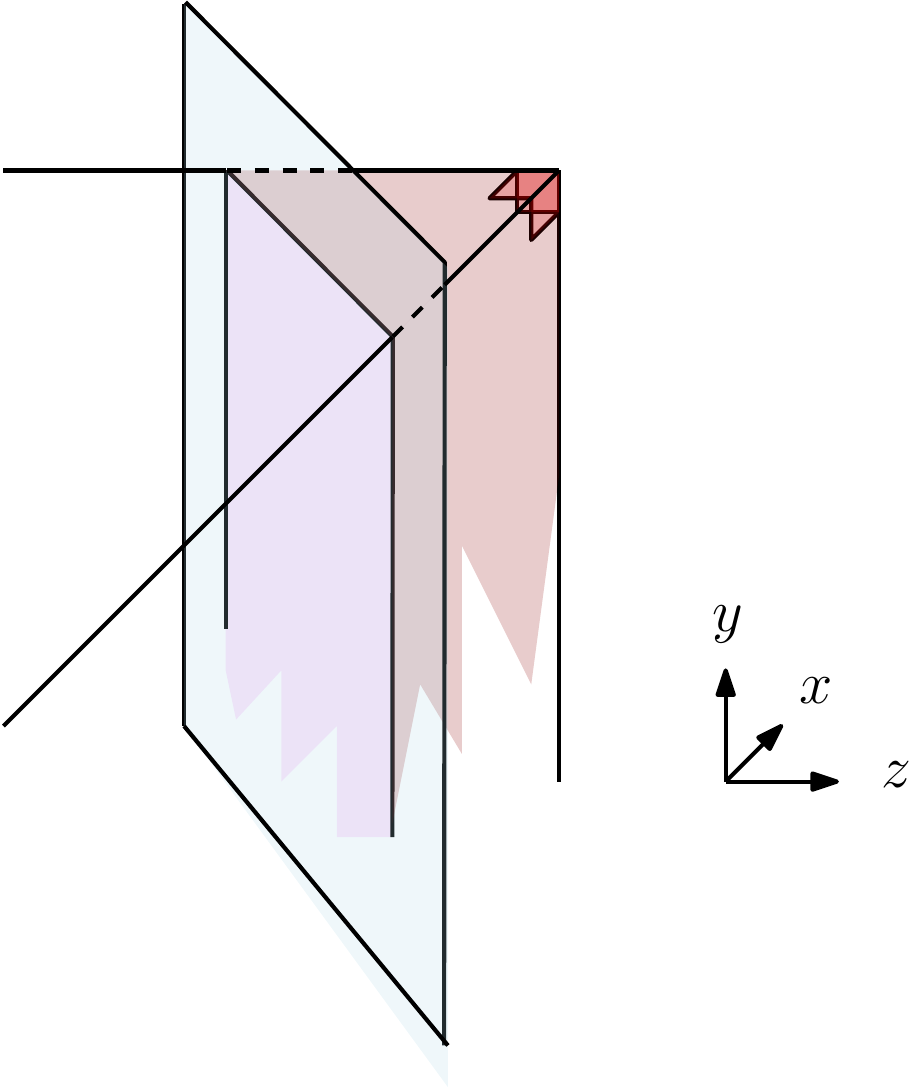}}
\end{center}
\caption{Special cases of the octant coloring problem.}
\end{figure}

\subsection{Intervals, bottomless rectangles, and sweeping algorithms.}

It is well-known that a theorem similar to Corollary~\ref{cor:inf} holds for the simpler case of intervals on the real line. Rado~\cite{R48} observed that every $k$-fold covering of the real line by intervals can be decomposed into $k$ coverings. Using this result, it is not difficult to prove that $p(k)=k$ for translates of negative quadrants in $\mathbb{R}^2$. 

In the second part of the paper, we study the problem of maintaining a decomposition of a set of intervals {\em under insertion}. This problem is similar in spirit, but distinct from the previous one. 

We are now given a positive integer $k$, a collection of intervals on the real line, and for each such interval a real number representing an {\em insertion time}. This collection represents a set of intervals that evolves over time, in which the intervals present at time $x$ are exactly those whose insertion time is at most $x$. We can now wonder whether there exists a function $p(k)$ such that the following holds: there exists a $k$-coloring of the intervals in the collection $\S$ such that, {\em at any time}, any point that is covered by at least $p(k)$ intervals {\em present at that time} is covered by at least one of each color.

This can be conveniently represented in the plane by representing each interval $[a,b]$ with insertion time $t$ as an axis-aligned rectangle with vertex coordinates $(a, -t), (b, -t), (b, -\infty ), (a, -\infty )$, hence viewing time as going downward in the vertical direction. We refer to such rectangles, with a bottom edge at infinity, as {\em bottomless} rectangles. Now the $k$-coloring must be such that every point $p\in \mathbb{R}^2$ that is contained in at least $p(k)$ such rectangles must be contained in at least one of each color. Hence the problem is actually about {\em decomposition of coverings by bottomless rectangles}. We illustrate this point of view in Figure~\ref{fig:bottomless}. Also note that bottomless rectangles can be seen as \emph{degenerate homothetic triangles}, which we will make use of for Corollary~\ref{cor:negative}. 

\begin{figure}
\begin{center}
\includegraphics[width=.4\textwidth]{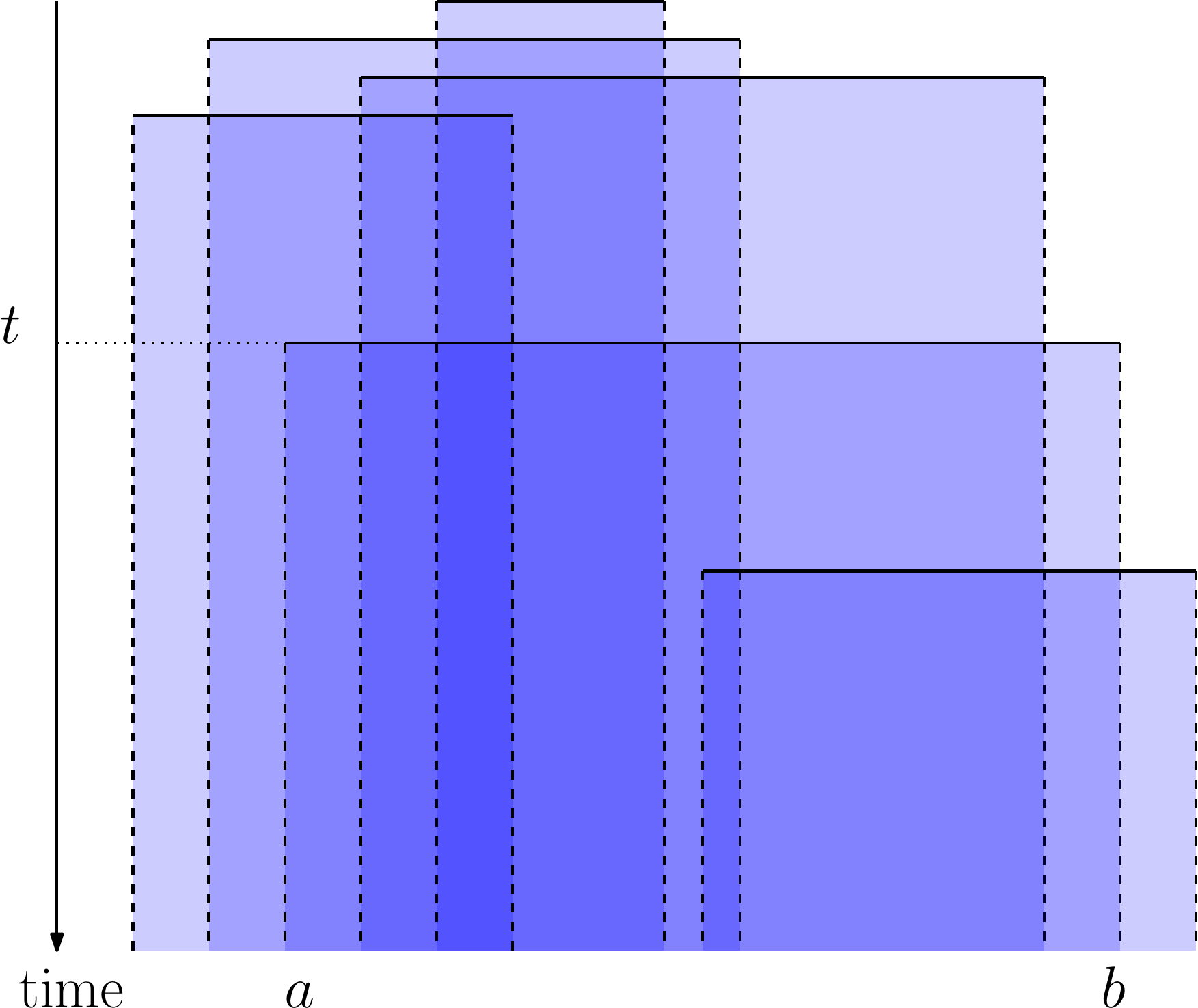}
\end{center}
\caption{\label{fig:bottomless}Intervals under insertion and bottomless rectangles.}
\end{figure}

We now observe that bottomless rectangles can be formed by intersecting a negative octant with a vertical plane, as depicted on Figure~\ref{fig:br}. Hence we can formulate a new corollary of our main Theorem.
\begin{corollary}
\label{cor:dualbott}
There exists a constant $a<6$ such that for any positive integer $k$, every finite set $\bigp$ of bottomless rectangles in the plane can be colored with $k$ colors so that every point contained in at least $k^a$ rectangles of $\bigp$ is contained in at least one of each color. Equivalently, every collection of intervals, each associated with an insertion time, can be $k$-colored so that at any time, every point covered by at least $k^a$ intervals present at this time is covered by at least one of each color.
\end{corollary}

With respect to the model of intervals with insertion times it is natural to ask whether it is possible to maintain a decomposition of a set of intervals under insertion, {\em without knowing the future insertions} in advance. In Section~\ref{sec:adversary}, we answer this question in the negative even if  coloring decisions can be delayed.

More precisely, we rule out the existence of a semi-online algorithm:

A \emph{semi-online $k$-coloring algorithm} must consider the intervals in their order of insertion time. At any time, an interval in the sequence either has one of the $k$ colors, or is left uncolored. Any interval can be colored at any time, but once an interval is assigned a color, it keeps this color forever. 

A semi-online $k$-coloring algorithm is said to be \emph{colorful} of \emph{value $d$} if it maintains at all times that the colors that are already assigned are such that any point contained in at least $d$ intervals is contained in at least \emph{one of each of the $k$ colors}.

In order to obtain that there is no  semi-online colorful coloring algorithm of bounded value, we prove a stronger statement about the less restrictive \emph{proper coloring} problem. We call a semi-online $k$-coloring algorithm \emph{proper} of \emph{value $d$} if it maintains at all times that the colors that are already assigned are such that any point contained in at least $d$ intervals is contained in at least \emph{two of distinct colors}. Our theorem says that for all natural numbers $k,d$, there is no semi-online proper $k$-coloring algorithm of value $d$.

\begin{theorem}
\label{thm:negative}
For all natural numbers $k,d$, there is no semi-online algorithm that $k$-colors intervals under the operation of inserting intervals, so that at any time, every point covered by at least $d$ intervals is covered by at least two of distinct colors.
\end{theorem}

Since any semi-online colorful coloring algorithm is also proper, we obtain that there is no such algorithm of bounded value. 

Note that in the bottomless rectangle model a  semi-online colorful coloring algorithm corresponds to sweeping the set of rectangles top to bottom with a line parallel to the $x$-axis and assigning colors irrevocably to already swept rectangles such that at any time every point contained in $d$ of the already swept ones is contained in at least one of each color. Similarly, one can define sweeping line algorithms for coloring homothetic triangles, where the point set is swept top to bottom by a line parallel to one of the sides of the triangles. For octants a sweeping plane algorithm would sweep the point set from top to bottom with a plane parallel to the $x,y$-plane. Since bottomless rectangles can be viewed as a special case of either we can summarize:

\begin{corollary}\label{cor:negative}
 For all natural numbers $k,d$, there is no sweeping line (plane) coloring algorithm in the above sense such that for any set of bottomless rectangles, or triangles, or octants, at any time every point contained in $d$ of the already swept ranges is contained in at least one of each color.
\end{corollary}

Since for octants primal and dual problem are equivalent by Corollary~\ref{cor:negative} no such sweeping plane algorithm exists for the primal octant problem either. 

We remark that Corollary~\ref{cor:negative} is in contrast with another recent result in~\cite{A12}, which deals with the primal version of the problem. It can be expressed as coloring points appearing on a line in such a way that at all times any interval containing $p(k)$ points contains one point of each color, or equivalently, coloring point sets in the plane such that every bottomless rectangle containing $p(k)$ points contains a point of each color. In~\cite{A12} it is shown that in this case a linear upper bound on $p(k)$ can be achieved with a semi-online coloring algorithm or equivalently a sweeping line algorithm.

\section{Previous Results}

The covering decomposition problem was first posed by J\'anos Pach in the years 1980-1986~\cite{Pa80,Pa86}. This was originally motivated by the problem of determining the densities of the densest $k$-fold packings and the thinnest $k$-fold coverings of the plane with a given plane convex body (see Section 2.1 in \cite{BMP} for a complete historical account). In particular, he posed the following problem:\\

{\it Is it true that for any plane convex polygon $C$ and for any integer $k$, there exists an integer $p=p(C,k)$ such that every $p$-fold covering of the plane with homothetic
copies of $C$ can be decomposed into $k$ coverings?}\\

\noindent
Our contribution shows that $p(C,k)=O(k^6)$ provided that $C$ is a triangle and the covering is locally finite (Corollary~\ref{cor:inf}). 

Tremendous progress has been made recently in understanding the conditions for the existence of a function $p(k)$ for a given range space, that is, geometric hypergraphs induced by a family of bodies in $\mathbb{R}^d$. To our knowledge, our result is the first polynomial bound for cover-decomposability of {\em homothetic} copies of a polygon. Linear upper bounds have been obtained for halfplanes~\cite{ACCLS09,SY10}, and translates of a convex polygon in the plane~\cite{TT07,PT09,ACCLOR10,GV11}. A restricted version of this problem involving unit balls is shown to be solvable using the probabilistic method in the well-known book from Alon and Spencer~\cite{AS}. The function $p(k)$ has been proved not to exist for range spaces induced by concave polygons~\cite{P10}, axis-aligned rectangles~\cite{CPST09,PT10b}, lines in $\mathbb{R}^2$, and disks~\cite{PTT05}. In a remarkable recent preprint, P{\'a}lv{\"o}lgyi proved the non-existence of a function $p(k)$ even for {\em unit} disks~\cite{P13}, thereby 
invalidating earlier claims from Pach and Mani-Levitska in an unpublished manuscript. Note that the indecomposability results for axis-aligned rectangles imply the same for orthants in $\mathbb{R}^4$, since arbitrary such rectangles can be formed by intersecting four-dimensional orthants with a plane in $\mathbb{R}^4$. 

Finally, in another recent preprint, Kov\'acs~\cite{Ko13} proved that there exists indecomposable coverings by homothets of any polygon with at least four edges, disproving the general conjecture above. Overall, this collection of results essentially closes most of J\'anos Pach's questions on cover-decomposability of plane convex bodies.

\subsubsection*{Previous results on octants.}

P{\'a}lv{\"o}lgyi proved the indecomposability of coverings by translates of a convex polyhedron in $\mathbb{R}^3$~\cite{P10}. His proof, however, does not hold for unbounded polyhedra with three facets. This prompted the first author of the current paper to pose the problem of decomposability of coverings by octants. This was solved by Keszegh and P{\'a}lv{\"o}lgyi, who showed that $p(2)\leq 12$ in this case~\cite{KP11}. Since we will reuse this theorem in our proof, it is worth reproducing it here.

\begin{theorem}[\cite{KP11}]
\label{thm:KP}
There exists a constant $c\leq 12$ such that every finite collection $\bigp\subset \mathbb{R}^3$ of points can be 2-colored so that every negative octant containing at least $c$ points of $\bigp$ contains at least one of each color.
\end{theorem}

In the past two years, the above result was improved and generalized. First, Keszegh and P{\'a}lv{\"o}lgyi proved that Theorem~\ref{thm:KP} implies that $p(k)$ is bounded for every $k$~\cite{KP14}. Note that this is not obvious, as one could well imagine that for some range spaces, $p(2)$ is bounded, but not $p(k)$ for some $k>2$. Their upper bound on $p(k)$, however, is doubly exponential in $k$. In particular, their proof implies $p(k)\leq 12^{2^k}$. 

Later, the current authors gave a polynomial upper bound on $p(k)$, but restricted to the special case of homothetic triangles in the plane, where points are to be colored~\cite{CKMU13}. The proof uses a new technique involving recoloring each color class of a $k$-coloring with two colors in order to obtain a $2k$-coloring. 

Finally, in May 2013, an unpublished manuscript from Keszegh and P{\'a}lv{\"o}lgyi was communicated to us by J\'anos Pach, in which an improved polynomial upper bound was given for the same special case of homothetic triangles~\cite{KP13}. This improvement makes use of a lemma stating the so-called {\em self-coverability} property of triangles.

We managed to harness the power of these observations for the general case of octants. In particular, we reuse the recoloring algorithm given from~\cite{CKMU13} in Lemma~\ref{lem:splitcol}, and also give a three-dimensional generalization of the self-coverability lemma of~\cite{KP13} in the form of Lemma~\ref{lem:sweep}. Our proof of Theorem~\ref{thm:ub} is longer than that of the doubly exponential upper bound in~\cite{KP14}, but not significantly more involved.

\subsubsection*{Previous results on online coloring problems and proper colorings of geometric hypergraphs.}

Semi-online algorithms have proved to be useful in an interesting special case of the problem with octants, in which all points considered in Theorem~\ref{thm:ub} lie on a vertical plane. This setting can be thought of as points appearing on a line, and we want to color the points with $k$ colors such that at any time, any set of $p(k)$ consecutive points contains at least one of each color. This problem has been studied by a number of authors, whose results were compiled in a joint paper~\cite{A12}. In particular, they showed that under this restriction, we have $1.6k\leq p(k)\leq 3k-2$. The upper bound is achieved using a semi-online algorithm, that does not require the knowledge of the future point insertions, and never recolors a point. This also amounts to coloring {\em primal} range spaces induced by bottomless rectangles with a sweeping line algorithm, i.e., coloring points such that bottomless rectangles containing many of them contain all colors. 

In contrast to our negative result about semi-online algorithms, a larger class of algorithms called \emph{quasi-online} has led to a short proof that $p(2)=3$ in the setting corresponding to our Corollary~\ref{cor:dualbott}, see~\cite{KLP13}, and is indeed also used to obtain Theorem~\ref{thm:KP} in~\cite{KP11}.

Clearly, colorful $2$-colorings and proper $2$-colorings coincide, but also for a larger number of colors proper colorings of geometric hypergraphs have been considered in the primal and dual setting. There are results for bottomless rectangles~\cite{K11}, halfplanes~\cite{F10,K11}, octants~\cite{CK13}, rectangles~\cite{CPST09,AEGSR12,PT10b}, and disks~\cite{PTT05,S07}.

Similarly to our Theorem~\ref{thm:negative} Keszegh, Lemons, and P{\'a}lv{\"o}lgyi consider \emph{online} proper coloring algorithms (points must be colored on arrival). While it is easy to see that there is an optimal online algorithm to color points such that quadrants are colorful, they show that there is no online proper coloring algorithm of bounded value in the primal setting of bottomless rectangles and octants. This is implied by Theorem~\ref{thm:negative} and indeed the proof methods have similarities. In~\cite{KLP13} the quality of online algorithms is then measured as a function of the input size. 

In another vein, Bar-Noy, Cheilaris, Olonetsky, and Smorodinsky~\cite{BCS08, BCOS10} considered {\em conflict-free} colorings in an online setting. There, the problem is to maintain that every $d$-covered point $p$ is covered by one interval whose color is unique among all intervals covering $p$.

\subsubsection*{Other related results.}

In 2010, Kasturi Varadarajan gave a feasibility result for the {\em fractional set cover packing} problem with fat triangles (Corollary 2 in \cite{V10}). This problem can be seen as a fractional variant of the covering decomposition problem. This result involves the construction of so-called {\em quasi-uniform $\varepsilon$-nets}. This construction was recently improved by Chan, Grant, K{\"o}nemann, and Sharpe~\cite{CGKS12}. These results are essentially motivated by the design of improved approximation algorithms for geometric versions of the weighted set cover problem. However, they can also be seen as an intermediate step between the problem of finding small $\varepsilon$-nets and the covering decomposition problem, which involves partitioning a set into $\varepsilon$-nets (see our conclusion for a discussion on this relation).

\section{Proof of Theorem \ref{thm:ub}}
\label{sec:pf1}

In what follows, we will use the shorthand notation $[n] = \{1,2,\ldots ,n\}$, for a positive integer $n$. 
We will refer to the three coordinates of a point $p$ as $p_x$, $p_y$, and $p_z$, respectively. 
The {\em negative octant} with {\em apex} $(p_x,p_y,p_z)\in\mathbb{R}^3$ is the set $\{(x,y,z)\in \mathbb{R}^3 : x\leq p_x, y\leq p_y, z\leq p_z \}$. 
Similarly the {\em positive octant} of $(p_x,p_y,p_z)$ is  $\{(x,y,z)\in \mathbb{R}^3 : x\geq p_x, y\geq p_y, z\geq p_z \}$. 
For convenience we also allow the coordinates of an apex to be equal to $\infty$. In what follows, an octant will generally be considered to be negative, unless explicitly stated otherwise. 
For two points $p,q\in \mathbb{R}^3$, we say that $p$ {\em dominates} $q$ whenever the negative octant with apex $p$ contains $q$, or, equivalently,
whenever $p$ is greater than $q$ coordinate-wise. 
We say that a set of points $\bigp\subset \mathbb{R}^3$ is {\em independent} whenever no point in $\bigp$ is dominated by another.
Finally, we say that a point set is in {\em general position} whenever no two points have the same $x$, $y$, or $z$-coordinates. 
By a standard perturbation argument it suffices to prove Theorem~\ref{thm:ub} for point sets in general position.

\begin{lemma}
\label{lem:sweep}
For every finite independent set $\bigp\subset \mathbb{R}^3$ in general position, there exists a collection $\bign$ of negative octants such that:
\begin{enumeratei}
\item $|\bign | = 2|\bigp | + 1$,
\item the octants in $\bign$ do not contain any point of $\bigp$ in their interior,
\item all points of $\mathbb{R}^3$ that do not dominate any point in $\bigp$ are contained in $\bigcup \bign$.
\end{enumeratei}
\end{lemma}
\begin{proof}
Let $n=|\bigp|$. 
We prove the lemma by induction on $n$. 
For $n=0$ we take the negative octant covering the whole space with apex $(\infty,\infty,\infty)$. 
If $\bigp =\{p\}$, then we take the octants with apices $(p_x,\infty,\infty)$, $(\infty,p_y,\infty)$, and $(\infty,\infty,p_z)$.
For $n\geq 2$ we consider the points of $\bign$ in order of increasing $z$-coordinates. 
Let us denote them by $p_1, p_2,\ldots ,p_n$ in this order. 
Note that since $\bigp$ is independent, we have $p_{i,x}<p_{j,x}$ or $p_{i,y}<p_{j,y}$ for every $i,j\in [n]$ such that $j<i$.

Suppose, for the sake of induction, that there exists such a collection $\bign_{n-1}$ for the first $n-1$ points of $\bigp$. 
We then consider the next point $p_n$ and construct a new collection $\bign_n$. We do this in three steps. First, we include in $\bign_n$ all the octants of $\bign_{n-1}$ that do not contain $p_n$. Then for each octant $Q'\in\bign_{n-1}$ such that $p_n\in Q'$, we let $Q$ be the octant having the same apex as $Q'$, but with its $z$-coordinate changed to $p_{n,z}$. We add each such octant $Q$ to $\bign_n$. 
Finally, we add two new octants to $\bign_n$. The first octant, $L_n$ (for {\em left}) will have the point $(p_{n,x}, y, \infty)$ as apex, where $y = \min (\{p_{j,y} : 1\leq j < i, p_{j,x} < p_{n,x}\} \cup\{\infty\})$. The second, $B_n$ (for {\em bottom}), will have the point $(x, p_{n,y}, \infty)$ as apex, where $x = \min (\{p_{j,x} : 1\leq j<i, p_{j,y} < p_{n,y} \}\cup\{\infty\})$. See Figure~\ref{fig:octantsweep2} for an illustration.

\begin{figure}[ht]
\begin{center}
\includegraphics[width=.4\textwidth]{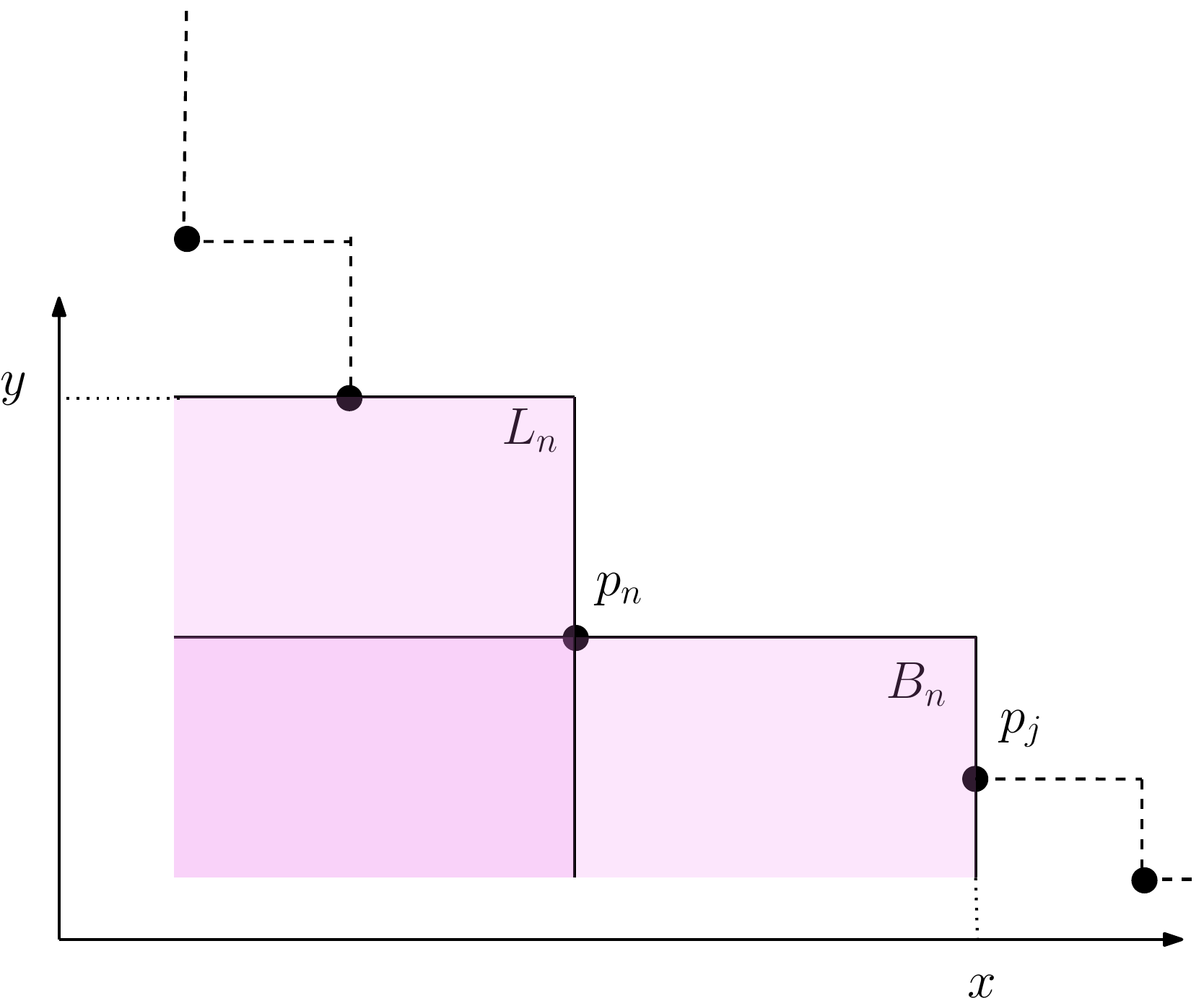}
\end{center}
\caption{Octants $L_n$ and $B_n$ in the proof of Lemma~\ref{lem:sweep}.}\label{fig:octantsweep2}
\end{figure}

The first property on the cardinality of $\bign_n$ holds by construction, as we add exactly two octants at each iteration.
The second property can be checked as follows. First, by the induction hypothesis, octants in $\bign_{n-1}$ avoid $p_1,\ldots,p_{n-1}$. 
Those octants from $\bign_{n-1}$ which avoid $p_n$ were copied to $\bign_n$ and others have their $z$-coordinate modified in a way to avoid $p_n$. 
Finally, the two new octants $L_n$ and $B_n$ have their interiors disjoint from $\bigp$ by definition and the fact that $\bigp$ is independent.

In order to verify the third property, let us consider a point $p'$ that is not dominating any point of $\bigp$. First suppose that 
$p'_z < p_{n,z}$. By induction, there exists an octant in $\bign_{n-1}$ containing $p'$. This octant is either contained in $\bign_n$, or has
its counterpart in $\bign_n$ with a modified $z$-coordinate. In both cases, $p'$ is covered by this octant in $\bign_n$. Now suppose that $p'_z\geq p_{n,z}$. 
We can further suppose that $p'$ neither belongs to $L_n$ nor to $B_n$. Then either $p'_x > x$, or $p'_y > y$, where $x$ and $y$ are the two values used to define $L_n$ and $B_n$. Let us suppose that $p'_x > x$, the other case being symmetric. Let $p_j, j<n$, be the point realizing the minimum in the definition of $x$. We must have $p'_y<p_{j,y}$, as otherwise $p'$ would dominate $p_j$. Then $p'$ must be covered by an octant $Q\in\bign_{n-1}$ whose $y$-coordinate is smaller than $p_{j,y}$, as otherwise $Q$ would contain $p_j$. But by definition $p_{j,y}<p_{n,y}$, hence $Q$ does not contain $p_n$ and therefore also belongs to $\bign_n$. In all cases, $p'$ is contained in an octant of $\bign_n$ and the third property holds.
\end{proof}

Note that the upper bound on the size of $\bign$ in Lemma~\ref{lem:sweep} is tight. For example, consider the point sets $\bigp_n = \{(i,-i,-i)\;|\; i =1,\ldots,n\}$. Indeed, Lemma~\ref{lem:sweep} and the fact that it is tight for all point sets that are in general postition and do not lie in a plane containing the all-ones vector can be deduced from a more general theorem of Scarf~\cite{FK08}.

In order to prove our main theorem, we will use Theorem~\ref{thm:KP}, due to Keszegh and P{\'a}lv{\"o}lgyi.
We proceed to describe a coloring algorithm that achieves the bound of Theorem~\ref{thm:ub}. We do this in two steps. First, we consider the case where the points to color form an independent set.

\begin{lemma}
\label{lem:splitcol}
Let $c$ be a constant satisfying the property in Theorem~\ref{thm:KP}. For any positive integer $k$, every finite independent set $\bigp\subset \mathbb{R}^3$ in general position can be colored with $k$ colors so that every negative octant containing at least $ck^{\log_2 (2c-1)}$ points of $\bigp$ contains at least one of each color.   
\end{lemma}
\begin{proof}
For $k=2$, we know there exists a 2-coloring of $\bigp$ satisfying the property of Theorem~\ref{thm:KP}. Suppose now, as an induction hypothesis, that we have a $k$-coloring $\phi$ of $\bigp$ such that every octant containing at least $p(k)$ points contains at least one of each color. Label the colors of $\phi$ by $1,\ldots,k$. 

We now describe a $2k$-coloring $\phi'$. For $i\in [k]$, let $\bigp_i=\phi^{-1}(i)$ be the set of points with color $i$. 
We know from Theorem~\ref{thm:KP} that there exists a 2-coloring $\phi_i:\bigp_i\to\{ i',i''\}$ of $\bigp_i$ such that every octant containing at least $c$ points of $\bigp_i$ contains at least one of each color $i'$ and $i''$. We now define $\phi'$ as the $2k$-coloring obtained by partitioning each color class in this way. We now claim that the coloring $\phi'$ is such that any octant containing at least $(2c-1) p(k)$ points contains at least one of each of the $2k$ colors. 

For the sake of contradiction, let $Q$ be an octant containing at least $(2c-1) p(k)$ points of $\bigp$, but not any point of color $i'$ in $\phi'$. Let $\bigp_Q\subseteq \bigp$ be the set of points contained in $Q$. If $Q$ does not contain any point of color $i'$, it means that it contains at most $c-1$ points of $\phi^{-1}(i)$. Let $\bigp_i = \phi^{-1}(i)\cap \bigp_Q$ be the points of color $i$ in $\phi$ contained in $Q$.

From Lemma~\ref{lem:sweep} and the fact that $\bigp_Q\subset\bigp$ is an independent set, we know that there exists a collection $\bign$ of at most $2(c-1)+1 = 2c-1$ octants whose interiors do not contain any point of $\bigp_i$, but that collectively cover all points of $\bigp_Q\setminus \bigp_i$. Indeed, after intersecting with $Q$ we can assume that $\bign$ covers precisely $\bigp_Q\setminus \bigp_i$ and no other point of $\bigp$.

Hence from the pigeonhole principle, one of the octants $N\in\bign$ contains at least $\lceil((2c-1) p(k)-(c-1)) / (2c-1)\rceil = p(k)$ points of $\bigp_Q$ in its interior, but no point of $\bigp_i$. From the general position assumption, we can find an octant contained in $N$ that contains \emph{exactly} $p(k)$ points of $\bigp_Q$, but no point of $\bigp_i$. But this is a contradiction with the induction hypothesis, since this octant should have contained a point of color $i$ in $\phi$.                

It remains to solve the following recurrence, with starting value $p(2)=c$:
\begin{eqnarray*}
p(2k) & \leq & (2c-1) p(k) \\
p(k) & \leq & c(2c-1)^{\lceil \log_2 k\rceil-1} \\
& < & ck^{\log_2 (2c-1)}
\end{eqnarray*}
\end{proof}

We now describe our algorithm for coloring an arbitrary set of points in general position. This requires a new definition.

Given a set $\bigp$ of points in general position in $\mathbb{R}^3$, the {\em minimal points} of $\bigp$ is the subcollection of points of $\bigp$ that are not dominating any other point of $\bigp$. In general, we define the {\em $i$th layer} $\bigl_i$ of $\bigp$ as its minimal points for $i=1$, and as the minimal points of $\bigp\setminus \bigcup_{1\leq j<i} \bigl_j$ for $i>1$. By definition each layer is an independent set of points.

%
\begin{lemma}
\label{lem:layeredcoloring}
Let $f(k) = ck^{\log_2 (2c-1)}$ be the function derived in Lemma~\ref{lem:splitcol}, where $c$ is a constant satisfying the property in Theorem~\ref{thm:KP}. For any positive integer $k$, every finite set $\bigp\subset \mathbb{R}^3$ in general position can be colored with $k$ colors so that every negative octant containing at least $(k-1)f(k)$ points of $\bigp$ contains at least one of each color.   
\end{lemma}
\begin{proof}
We will color the points of $\bigp$ by considering the successive layers one by one, starting with the minimal points. For each layer $\bigl_i$, we do the following:
\begin{itemize}
\item precolor the points of $\bigl_i$ with colors in $[k]$, as is done in Lemma~\ref{lem:splitcol},
\item for each point $p\in \bigl_i$:
\begin{itemize}
\item consider the set of points $D_p=\{q\in\bigp: q \textrm{ dominated by } p\}$; 
\item if $p$ is precolored with a color that is not used for any point in $D_p$ then this color is the final color of $p$;
\item otherwise pick any color not present on points in $D_p$ and color $p$ with it; if all $k$ colors are used within $D_p$, leave $p$ uncolored.
\end{itemize}
\end{itemize}
The main observation here is that although the recoloring step harms the validity of the coloring within a single layer, it is globally innocuous, since any octant containing the point $p$ in the $i$th layer also contains all the points in $D_p$, from the previous layers. Thus, any octant containing $p$ contains a point colored by the same color as the precolor of $p$. Note that each point in the $i$th layer must dominate at least one point from each $i-1$ earlier layers. This forces the invariant that any octant containing a point of the $i$th layer contains points with at least $i$ distinct colors. In particular, any octant containing a point of the $k$th layer will contain all the colors.

The analysis is now straightforward. Suppose that an octant contains at least $(k-1)f(k)$ points. If it contains a point of the $k$th layer, then it contains all $k$ colors. Otherwise, it must contain points of at most $k-1$ layers, and from the pigeonhole principle, it contains at least $(k-1)f(k)/(k-1)=f(k)$ points in a single layer. Then the precoloring of this single layer guarantees each octant of size at least $f(k)$ to be colorful.
\end{proof}

Now Theorem~\ref{thm:ub} follows by replacing $c$ by 12 in the expression of Lemma~\ref{lem:layeredcoloring}, yielding $a\simeq 5.58$.

\section{Proof of Theorem~\ref{thm:negative}}
\label{sec:adversary}

\begin{figure*}[htb]
 \centering
 \includegraphics{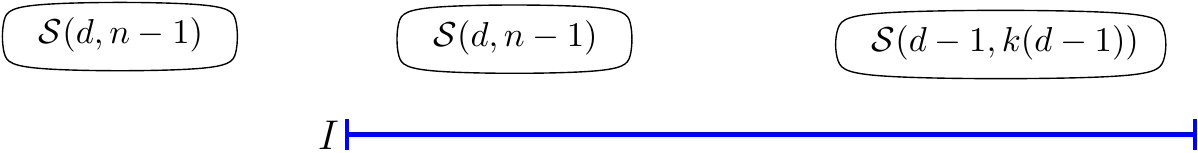}
 \caption{Defining strategy $\S(d,n)$ once $\S(d-1,k(d-1))$ and $\S(d,n-1)$ are defined, in the case where $t_i=t'_i$ for all $i \in [k]$.}
 \label{fig:negative}
\end{figure*}

\begin{proof} We say that a point of the real line is \emph{$d$-covered}, if it is contained in exactly $d$ intervals presented so far. We shall define for every $d$ and $n$ an adversarial strategy $\S(d,n)$ for presenting intervals such that the following is true:

 \begin{enumeratei}
  \item Every semi-online proper $k$-coloring algorithm of value at most $d$ executed against $\S(d,n)$ yields $k$ points $p_1,\ldots,p_k$ such that for every $i\in [k]$, the point $p_i$ is eventually covered by exactly $t_i$ intervals, all of which have color $i$, and\label{enum:covered-points}
  \item $t_1+\ldots+t_k\geq n$.\label{enum:inequality}
 \end{enumeratei}

 Clearly, if for some  semi-online $k$-coloring algorithm $\mathcal{A}$ there is a point eventually covered by at least $d$ intervals, all of which have the same color, then the value of $\mathcal{A}$ is at least $d+1$. Thus if $S(d,kd)$ exists and satisfies~\ref{enum:covered-points} and~\ref{enum:inequality}, then there is no semi-online $k$-coloring algorithm of value at most $d$, which proves the theorem.

 \smallskip

 We prove the existence of $\S(d,n)$ by a double induction on $d$ and $n$. Strategies $\S(d,0)$ are vacuous as~\ref{enum:covered-points} and~\ref{enum:inequality} for $n=0$ hold for the empty set of intervals and any set of $k$ distinct points $p_1,\ldots,p_k$. We define $\S(d,n)$, for $n>0$, once we have defined $\S(d-1,k(d-1))$ and $\S(d,n-1)$. 

Before continuing let us present the following useful claim.

\noindent{\bf Claim.} Consider a set $\mathcal{I}$ of intervals already presented, $I\in\mathcal{I}$, and $I'\notin\mathcal{I}$ such that $I'\subset I$ and $I'\cap J=\emptyset$ for all $J\in\mathcal{I}\setminus I$. If $\S(d-1,k(d-1))$ exists we can present the intervals of $\S(d-1,k(d-1))$ inside $I'$ forcing any semi-online algorithm of value at most $d$ to color $I$.
\begin{proof}
We present the intervals for $\S(d-1,k(d-1))$ completely inside $I'$. If the algorithm does not color $I$ then it can be seen as a $k$-coloring algorithm of value at most $d-1$ executed against $\S(d-1,k(d-1))$. We already know that there is no such algorithm and therefore every $k$-coloring algorithm of value at most $d$ has to color interval $I$.
\end{proof}

Now, we are ready to define $S(d,n)$ for $n>0$. First present two families of intervals, both realizing strategy $S(d,n-1)$, disjointly next to each other. By~\ref{enum:covered-points} there exist two sets of $k$ points each, $p_1,p_2,\ldots,p_k$ and $p'_1,\ldots,p'_k$, and non-negative integers $t_1,\ldots,t_k,t'_1,\ldots,t'_k$ such that $p_i$ is $t_i$-covered and all its intervals are colored with $i$ and also $p'_i$ is $t'_i$-covered and all its intervals are colored with $i$, for every $i\in [k]$. Moreover, by~\ref{enum:inequality} we have $t_1+\ldots+t_k\geq n$ and $t'_1+\ldots+t'_k \geq n$.

If there exists some $i \in \{1,\ldots,k\}$ with $t_i\neq t'_i$ then the sequence of maxima $m_i=\max(t_i,t'_i)$ satisfies $m_1+\ldots+m_k\geq n+1$. Thus, taking for each $i \in \{1,\ldots,k\}$ the point from $\{p_i,p'_i\}$ that corresponds to the larger value of $t_i,t'_i$, we obtain a set of $k$ points satisfying~\ref{enum:covered-points} and~\ref{enum:inequality}.
 
Hence we assume without loss of generality that $t_i=t'_i$ for all $i \in \{1,\ldots,k\}$. Then we present one additional interval $I$ that contains all the points $p'_1,\ldots,p'_k$ but none of the points $p_1,\ldots, p_k$. Moreover, $I$ is chosen big enough so that there exists some $I' \subset I$ that is disjoint from all the other intervals presented so far. We present the intervals realizing strategy $S(d-1,k(d-1))$ inside $I'$, forcing $I$ to be colored (see Figure~\ref{fig:negative}). Let $j$ be the color of $I$. Then $p'_j$ is now contained in exactly $t'_j + 1$ intervals all of which are colored with $j$. Thus $(\{p_1,\ldots,p_k\} \setminus p_j) \cup \{ p'_j\}$ is a set of $k$ points satisfying~\ref{enum:covered-points} and~\ref{enum:inequality}, which concludes the proof.
\end{proof}

\section*{Discussion and Open Problems}

A well-studied problem in discrete geometry is to identify properties of range spaces, or geometric hypergraphs, that allow one to find small {\em $\varepsilon$-nets}. It is known, for instance, that range spaces of bounded VC-dimension have $\varepsilon$-nets of size $O(\frac1{\varepsilon}\log\frac1{\varepsilon})$ (See for instance Chapter 10 in Matou\v{s}ek's lectures~\cite{Mat}). 

The coloring problem that we consider can be cast as the problem of partitioning a point set into $k$ $\varepsilon$-nets for $\varepsilon = p(k)/n$. In fact, it is one of the negative result for covering decomposition that formed the basis of a construction of Pach and Tardos for proving superlinear lower bounds on the size of $\varepsilon$-nets~\cite{PT11}. One can realize that if $p(k)=O(k)$ for a given range space, it implies that this range space also has $\varepsilon$-nets of size $O(1/\varepsilon)$. The latter is known to hold for range spaces induced by octants~\cite{CV07}. Whether $p(k)=O(k)$ for octants is therefore an interesting open problem. In general, giving improved upper or lower bounds on $p(k)$ for octants is the major remaining open question.

Another interesting open question concerns the primal problem, in which points are colored with $k$ colors so that every region containing $p(k)$ points contains a point of each color. The existence of such a function $p(k)$ is still open for homothetic copies of a square, for instance.

\subsection*{Acknowledgments.}

The authors wish to thank Bal{\'a}zs Keszegh, J\'anos Pach and D{\"o}m{\"o}t{\"o}r P{\'a}lv{\"o}lgyi for stimulating discussions on this topic.

Parts of this work were carried out during a visit of Kolja, Piotr, and Torsten at ULB in Brussels, funded by a EUROCORES Short Term Visit grant, and during a visit of Jean at the Jagiellonian University in Krak\'ow.

Research is supported by the ESF EUROCORES programme EuroGIGA, CRP ComPoSe and GraDR. Jean Cardinal is supported by the F.R.S-FNRS as grant convention no.\ R 70.01.11F within the ComPoSe project. Piotr Micek is supported by the Ministry of Science and Higher Education of Poland as grant no.\ 884/N-ESF-EuroGIGA/10/2011/0 within the GraDR project.
Kolja Knauer is supported by the ANR TEOMATRO grant ANR-10-BLAN 0207 and the EGOS grant ANR-12-JS02-002-01. 

\bibliographystyle{plain}
\bibliography{biblio_p(k)_vs_c(k)}

\end{document}